\documentclass[11pt]{article}
\usepackage{fullpage}
\usepackage{times}
\usepackage{amsmath,amsfonts,amssymb,amscd}
\usepackage{amsthm}
\usepackage{mymacros} 

\newtheorem{definition}{Definition}[section]
\newtheorem{theorem}[definition]{Theorem}

\newtheorem{lemma}[definition]{Lemma}
\newtheorem{observation}[definition]{Observation}

\newtheorem{corollary}[definition]{Corollary}

\DeclareMathOperator{\val}{value}
\DeclareMathOperator{\valu}{val}
\DeclareMathOperator{\totval}{tval}
\DeclareMathOperator{\tv}{tval}

\DeclareMathOperator{\lmove}{lem}

\begin{document}

\title{Certifiably Pseudorandom Financial Derivatives\thanks{An extended abstract of this paper appeared in EC'11.
}
}

\author{
David Zuckerman 
\\
Department of Computer Science\\
University of Texas at Austin\\
2317 Speedway, Stop D9500\\
Austin, TX  78712\\
{\tt diz@cs.utexas.edu}
}

\maketitle
\begin{abstract}
Arora, Barak, Brunnermeier, and Ge \cite{abbg} showed that taking computational
complexity into account, a dishonest seller could strategically place lemons in financial derivatives to make them substantially less valuable to buyers.
We show that if the seller is required to construct derivatives of a certain form, then this phenomenon
disappears.
In particular, we define and construct \emph{pseudorandom derivative families}, for which lemon placement only slightly affects the values of the derivatives.
Our constructions use 
expander graphs.

We study our derivatives in a
more general setting than Arora et al.
In particular, we analyze arbitrary tranches of the common 
collateralized debt obligations (CDOs) when the underlying assets can have significant dependencies.
\end{abstract}




\newcommand{\assetsize}{n} 
\newcommand{\derivsize}{m} 
\newcommand{\leftsize}{n}
\newcommand{\rightsize}{m}
\newcommand{\lemonsize}{\ell} 
\newcommand{\badsize}{k} 
\newcommand{\ksize}{\ell} 
\newcommand{\kmax}{\ksize_{max}}
\newcommand{\lmax}{{\ell_{max}}}
\newcommand{\degree}{d} 
\newcommand{\degreer}{r}
\newcommand{\leftdeg}{d}
\newcommand{\rightdeg}{r}
\newcommand{\assets}{A}
\newcommand{\derivs}{B}
\newcommand{\adegree}{\leftdeg}
\newcommand{\derivdeg}{\rightdeg}
\newcommand{\expansion}{\gamma} 
\newcommand{\badthresh}{t}
\newcommand{\rdegree}{d_r}
\newcommand{\globe}{Z}

\newcommand{\valab}{\val_{[a,b]}}
\newcommand{\valuab}{\valu_{[a,b]}}
\newcommand{\lmoveab}{\lmove_{[a,b]}}

\section{Introduction}
\label{sec:intro}

Financial derivatives play a major role in our financial system, as became all too apparent in the 2008 financial crisis.
A derivative is a financial product whose value is a function of one or more underlying assets.
They can be used to hedge risk, provide leverage, or simply to speculate.  The major benefit of derivatives is
that they facilitate the buying and selling of risk.

While a derivative may depend on only one asset, in this paper we study derivatives that depend on many assets,
specifically, collateralized debt obligations, or CDOs.  Over \$500 billion of CDOs were issued in 2006, but CDOs played a major role in the financial crisis.
Issuance plummeted after the financial crisis but in 2014 over \$100 billion of CDOs were issued, and the numbers continue to grow \cite{Buc}.

The structure of a CDO is quite intuitive.
A CDO packages many underlying assets into \emph{tranches}.
For example, a CDO could have 100 underlying mortgages, each of which is supposed to pay \$1,000.
The ``senior" tranche, for instance, could collect the first \$85,000.  Thus, if more than \$85,000 is paid from these 100 mortgages,
this tranche receives \$85,000; if some amount $x\leq \$85,000$ is paid, the tranche receives~$x$.
The next tranche could range from \$85,000 to \$95,000.  If more than \$95,000 is paid, this tranche
receives the full \$10,000; if less than \$85,000 is paid, this tranche receives nothing.  If 
the amount $x$ paid is between \$85,000 and \$95,000, then
the tranche receives $x - \mbox{\$85,000}$.
In general, the $[a,b]$ tranche receives $\min(x,b) - \min(x,a)$.
Note that senior tranches can be much safer than the underlying assets, which is part of the lure.

This paper focuses on such derivatives in the context of asymmetric information.
Akerlof introduced a framework to study such asymmetric information, using the market for used cars to capture the main ideas \cite{Ake}.
Here we follow the simplified discussion in Arora et al.\ \cite{abbg}.
There is an information asymmetry in that a seller knows whether a car is a lemon, i.e., functions poorly and is worth nothing, whereas a buyer cannot detect this with a short test drive.
Imagine that a functional used car is worth \$1000, but everybody knows that 20\% of used cars are lemons.
A buyer would be willing to pay at most \$800 for a car that he thought had a 20\% chance of being a lemon, but a seller who knows that her car is not a lemon would not be willing to sell it for this price.  In order for a sale to take place, the buyer's value must be \$200 more than the seller's value, and this is called the ``lemon cost."

Let's examine how asymmetric information affects CDOs.
A seller may be aware that certain underlying assets are lemons, and try to strategically place the lemons among the derivatives
in order to minimize the derivatives' value.
However, many believed that since the seller usually retains the junior tranches which take the first losses, the senior tranches are less vulnerable to such manipulation, giving smaller lemon cost.  In fact, DeMarzo proved this \cite{DeM}); however, he implicitly assumed that buyers had unlimited computational power.

Arora, Barak, Brunnermeier, and Ge \cite{abbg} introduced computational complexity into this discussion.
They showed that contrary to the conventional wisdom, once computational
complexity is accounted for, the lemon costs of derivatives could
increase dramatically, at least under a plausible computational assumption.

In the 2008 financial crisis, lemon packing did occur.
For example, in Chapter 4 of \emph{The Big Short} \cite{Lew},
Michael Lewis describes how Wall Street firms packed CDOs with loans to people with low FICO scores, or whose high FICO scores were based on a short credit history.
This is because the rating agencies Moody's and S\&P only asked for the average FICO score in a CDO.
A CDO of loans with half FICO scores of 550 and half with 680 is much riskier than one with all FICO scores of 615, because a FICO score of 550 indicates very likely default.
Moreover, the rating agencies ignored the length of the credit history.
Of course, this is not difficult to catch computationally, but it does indicate a willingness of banks to pack lemons into CDOs in a way that's hard to detect by a buyer.

Before describing how we get around the negative result of Arora et al.,
we first describe their model.
There are $\assetsize$ assets and $\derivsize$ derivatives, where each derivative is a function of $\derivdeg$ underlying assets.
We will have $\assetsize \ll \derivsize \derivdeg$, so that each asset underlies several derivatives.
This is typically not the case if the underlying assets are, say, mortgages; however, in the common case that the underlying assets are credit
default swaps, it is often the case that an asset underlies several derivatives.
Alternatively, we can replace each asset with an asset class of similarly performing assets, such as mortgages from the same market, and now duplicates correspond to samples from the same asset class.  Arora et al.\ go back and forth between assets and asset classes; for simplicity we stick to assets.

Arora et al.\ model the relationship between derivatives and underlying assets as a bipartite graph.
The nodes are the derivatives and assets, and there is an edge between a derivative and an asset if the derivative depends
on that asset.  If the derivatives are for sale to the public, then the seller must make this graph public.

Now consider a seller who knows that particular underlying assets are lemons.
For certain tranches of CDOs, it is advantageous for the seller to concentrate many of these lemons into a small number of derivatives.
Thus, these lemons and the lemon-loaded derivatives will correspond to a dense subgraph of the original graph.
Arora et al.\ observed that if it is computationally intractable to check whether an arbitrary graph (or even a somewhat random graph)
contains a dense subgraph, then
it is computationally intractable to catch such a dishonest seller.  Therefore, the lemon cost could be quite high.

Specifically, suppose there is no polynomial-time algorithm that distinguishes between random graphs and random graphs with a planted dense subgraph.
Say there are $\ell$ lemons.
Arora et al.\ define the lemon cost as the value without any lemons minus the value with lemons.
They show that while the lemon cost for a rational, time-unbounded buyer who can distinguish the above graphs
is $o(\ell)$, the lemon cost for a polynomial-time bounded buyer will actually be $\omega(\ell)$.

We circumvent this problem.
Instead of allowing the seller to use an arbitrary bipartite graph to construct the CDO family,
we mandate that the seller use a specific bipartite graph.
Of course, the buyer can easily check that the seller did use the specified graph.
While the seller will still be able to assign assets to nodes arbitrarily, we can choose the graph judiciously to avoid the dense subgraph problem,
for the following reason.
Although it may be computationally intractable to test whether an arbitrary graph,
or even a somewhat-random graph,
contains a dense subgraph, it is nevertheless possible to explicitly construct a graph with no dense subgraphs.
We choose such a graph for the seller.

Graphs with no dense subgraphs are related to certain fundamental objects in the theory of pseudorandomness:
randomness extractors and expander graphs.
For a survey of these objects and other aspects of pseudorandomness, see \cite{v:prg-survey}.

The reader may wonder who is requiring the use of such a graph.
Since we show that a seller can't gain much by strategically placing lemons for such a graph, a seller may be motivated to use such a graph to attract buyers.
Alternatively, a regulator might mandate the use of such a graph.
Financial authorities have tried to encourage the transparency of financial products since the 2008 crisis, and our CDO constructions could help contribute towards this worthy goal.

Our constructions motivate the notion of what we call a \emph{pseudorandom derivative family}.  This is a set of derivatives such that no matter how the
lemons are placed by an adversary, the sum of the value changes of the derivatives will be small.
In other words, adversarial placement of lemons behaves similarly to random placement.

Alternatively, we may decompose
the lemon cost into
the unavoidable lemon cost plus the cost of dishonest placement.
The unavoidable lemon cost is the lemon cost for an honest seller who randomly places the lemons.
The cost of dishonest placement is the additional cost from a dishonest seller who strategically places the lemons.
A~pseudorandom derivative family is one where the cost of dishonest placement is small.

One could imagine a few methods to force the seller to package the CDOs randomly.  For example, the seller could use a cryptographic hash function of the current time.  However, this still allows the seller some flexibility, in that the seller may choose various different times, and even different naming of the assets, to find the most profitable lemon placement.  Alternatively, the seller could use some fixed seemingly-pseudorandom string, such as the digits of $\pi$, but there is no guarantee that the digits of $\pi$ will have the desired property.

In our main result, we show how to construct good pseudorandom derivative families, using
expander graphs with expansion close to the degree.
Moreover, this is \emph{certifiably} pseudorandom, in the sense that there is a proof that the CDO packaging is close to fair.
In other words, there is zero chance of an unfair packaging, whereas with a cryptographic method there would be a positive chance of an unfair packaging.

Of course, in order to analyze values we need a model for the underlying assets.  Arora et al.\ assume the underlying
assets are independent fair coin flips, taking the value 1 with probability 1/2, whereas the lemons always take value 0.
As their result was negative, a simpler model gives a stronger result.  However, we strive for a positive result, so we
analyze a more realistic model with dependencies.

We use the factor framework, which is a common way to model CDOs (see \cite{CouL} for an overview).
All dependencies among the assets occur through a global random variable~$\globe$ that represents a set of ``factors."
For example, these factors could include the the state of the economy and housing market. 
Many papers are written where there is just one 1-dimensional factor, but our model is more general in that we make no assumptions about $\globe$.
Conditioned on $\globe$, all assets are independent.  

Furthermore, we only require that the probability distribution on any $\derivdeg$ assets depends solely on how many of the assets are lemons.
We also don't need to assume lemons have value 0.  Rather, for our strongest results, we assume that for any fixing of the global random variable,
good assets first-order stochastically dominate lemons (see Subsection~\ref{model}).

With this as our model, we study arbitrary tranches of CDOs.
We obtain an error bound of $2\Delta/d$ times the trivial bound, where $d$ is the left degree and small sets on the left expand by a factor of $d-\Delta$.
There are explicit constructions where this is $o(1)$.
We obtain even stronger results for the entire CDO.
Arora et al.\ analyze only senior tranches of CDOs.


\smallskip
\noindent {\bf Related Work.}
After the preliminary version of this paper, there has been more related work on the computational complexity of financial products.
Braverman and Pasricha showed that pricing compound options (options of options) can be computationally intractable, in fact PSPACE-hard \cite{BraP}.
Hemenway and Khanna proved that it is computationally intractable to estimate the number of failures caused by a small shock to a financial system \cite{HemK}.
Schuldenzucker, Seuken, and Battiston showed how the computational complexity of clearing financial networks can be greatly increased once banks enter into credit default swaps \cite{ssb}.
We believe there is a lot more to be explored at the intersection of finance and the theory of computation.

\smallskip
\noindent {\bf Organization.}
We begin by explaining our model and defining key terms in Section~\ref{sec:def}.
We then describe how expander graphs give pseudorandom CDOs in Section~\ref{sec:expand-cdo}.
We modify existing expander constructions to obtain our CDOs in Section~\ref{sec:explicit}.
We analyze the case when good assets don't necessarily stochastically dominate lemons in Section~\ref{sec:general}.
We discuss some extensions in Section~\ref{sec:extend}, and conclude in Section~\ref{sec:conclude}.

\section{The Model and Key Definitions
}
\label{sec:def}

First we give some notation.  For a positive integer $n$, we let $[n]$ denote the set $\set{1,2,\ldots,n}$.  For a vector $v=(v_1,\ldots,v_s)$,
we let $\|v\|_1=\sum_i |v_i|$ be the $L_1$ norm.

Our CDOs will be functions of underlying assets.  We first describe our assumptions about the underlying assets, and then define pseudorandom CDOs.

\subsection{Model for Underlying Assets}
\label{model}

In our model, there are two types of assets, \emph{lemons} and \emph{good assets}.
Good assets must first-order stochastically dominate lemons in a sense below.  This requirement will be satisfied if lemons always take value zero, but it allows more general distributions
on lemons.
Each CDO will depend on $\derivdeg$ assets.
Our results hold as long as the probability distribution on any $\derivdeg$ assets depends only on how many of the assets are lemons.

We now elaborate on one natural model which satisfies the two requirements above.
We model dependencies among assets as occurring through some global random variable~$\globe$.
This $\globe$ represents a set of factors, such as the state of the economy and housing market.
We make no assumptions about~$\globe$.
For each fixing of $\globe$, to say $z$,
there are two probability distributions $D_g = D_g(z)$ and $D_\ell = D_\ell (z)$.
Conditioned on $\globe=z$, our model assumes all assets are independent, with
good assets chosen according to $D_g$, and lemons chosen according to~$D_\ell$.
Moreover, we say good assets first-order stochastically dominate lemons if for any $z$ and $a$,
\[ \Pr_{X \sim D_g(z)}[X \geq a] \geq \Pr_{Y \sim D_\ell(z)}[Y \geq a].\]

We can relax the requirement that assets are conditionally independent.
It suffices that the conditional distribution on assets is $\rightdeg$-wise independent, i.e., any $\rightdeg$ of them are independent.
(This does not imply that they are mutually independent.)

We normalize asset values so that each asset's maximum value is 1.
We let $\mu$ and $\lambda$ be the expected values of each good asset and lemon, respectively.  The dominance requirement implies $\mu \geq \lambda$,
and let $\delta = \mu - \lambda$ be the additional expected value of a good asset.

\subsection{Pseudorandom CDOs}

\begin{definition}
\label{cdo-def}
A collateralized debt obligation (CDO) is a derivative on an underlying portfolio of 
assets.  For $0=a_0 < a_1 < \ldots < a_s$ 
(called attachment points),
the $i$th tranche is given by the interval $[a_{i-1},a_i]$.  If the underlying portfolio pays off $x$, then the value of the $i$th tranche is
$\val_{[a_{i-1},a_i]}(x)=\min(x,a_i) - \min(x,a_{i-1})$.
If the tranche is understood, we often omit it as a subscript in $\val$.
\end{definition}

Since assets are normalized to have maximum value 1, if the CDO depends on $\derivdeg$ assets, the last attachment
point is $a_s = \derivdeg$.

We will be interested in families of CDOs.

\newcommand{\cdofam}{$(\assetsize,\derivsize,\derivdeg)$-CDO family}
\begin{definition}
An $(\assetsize,\derivsize,\derivdeg)$-\emph{CDO family} is a set of $\derivsize$ CDOs on $\assetsize$ assets identified with the set $[\assetsize]$,
where each CDO depends on $\derivdeg$ assets.
\end{definition}

We will have $\assetsize \ll \derivsize \derivdeg$, so that each asset underlies several derivatives.

The seller (creator of the CDOs) knows that some $\lemonsize$ assets are lemons, and may identify the lemons with any
subset $L \subseteq [\assetsize]$ of size $\lemonsize$.  We will be interested in the total value of tranches in our CDO family.

\newcommand{\vectotval}{\vec{\tv}}
\begin{definition}
For $L \subseteq [\assetsize]$, let $\totval_{[a,b]}(L)$ denote the total expected value of all $[a,b]$ tranches in the CDO
family, if the assets corresponding to assets $L$ are lemons.
If the tranche is understood, we often omit it as a subscript.
We define the vector $\vectotval(L) = (\totval_{[a_0,a_1]}(L),\totval_{[a_1,a_2]}(L),\ldots,\totval_{[a_{s-1},a_s]}(L))$.
\end{definition}

A dishonest seller will try to choose the subset $L$ to minimize $\totval(L)$.
For example, Arora et al.\ assume that the seller retains all junior tranches, to signal that his assets are high quality.
He then has an incentive to concentrate risk in some senior tranches, minimizing the value of these tranches.
A CDO family is pseudorandom if the seller cannot gain significantly
by this choice.  In the scenario envisioned by Arora et al., we want pseudorandomness with respect to the senior tranches.

\begin{definition}
An $(\assetsize,\derivsize,\derivdeg)$-CDO family is \emph{pseudorandom} for $\lemonsize$ lemons for $[a,b]$ tranches
with error $\eps$ if for any two subsets $L,L' \subseteq [\assetsize]$ of size $\lemonsize$,
\[ |\totval_{[a,b]}(L') - \totval_{[a,b]}(L)| \leq \eps \derivsize(b-a). \]

\end{definition}

Note that $\derivsize(b-a)$ is the maximum possible value of the $[a,b]$ tranches with no lemons.
Thus, for any CDO family the error $\eps$ is at most~1.

We further define pseudorandomness for the entire CDO family.
We can't generalize the above definition naively, to say that the total value of the CDO doesn't change significantly if the lemons are moved.
This is because the total value of the CDO equals the total value of the underlying assets; therefore moving lemons won't change the value at all.
Instead, we strengthen the definition to ensure that not much value can be transferred among the different tranches.
That is, we add up the value changes of each tranche; this gives the $L_1$-norm.

\begin{definition}
An $(\assetsize,\derivsize,\derivdeg)$-CDO family is \emph{pseudorandom} for $\lemonsize$ lemons
with error $\eps$ if for any two subsets $L,L' \subseteq [\assetsize]$ of \mbox{size $\lemonsize$,}
\[ \|\vectotval(L') - \vectotval(L)\|_1 \leq \eps \derivsize \derivdeg. \]
\end{definition}

Note that $\derivsize \derivdeg$ is the maximum possible value of the entire CDO family with no lemons.
The error $\eps$ for the CDO family is at most the maximum error for a tranche, and hence at most 1.

We can compare our notion of pseudorandom error to the traditional notion of lemon cost.
The lemon cost is the value without any lemons minus the value with lemons.
In a sense, we are dividing the lemon cost into two components:  the unavoidable lemon cost plus the cost of dishonest placement.
The unavoidable lemon cost is the lemon cost for an honest seller who randomly places the lemons.
The cost of dishonest placement is the additional cost from a dishonest seller who strategically places the lemons.
Thus, the normalized cost of dishonest placement is upper bounded by the pseudorandom error.

On the other hand, the pseudorandom error above is at most the normalized lemon cost, but it could be significantly less.
For example, if all underlying assets are lemons, the lemon cost will be high, but the error in our definition will be 0, since the value doesn't change depending on the lemon placement.
Indeed, the pseudorandom error is small if the lemon cost doesn't depend significantly on the lemon placement.

\subsection{Bipartite Expander Graphs}

Following Arora et al., we view the relationship between derivatives and underlying assets as a bipartite graph.  We review the basic definitions.

\begin{definition}
A bipartite graph is a triple
$(\assets,\derivs,E)$, with \emph{left vertices} $\assets$, \emph{right vertices} $\derivs$, and \emph{edges} $E \subseteq \assets \times \derivs$.  We usually view $E$ as unordered pairs of vertices.
Sometimes we refer to a bipartite graph on $\assets \cup \derivs$ to mean some bipartite graph $(\assets,\derivs,E)$ with suitable choice of edges $E$.
For a subset of vertices $S \subseteq \assets \cup \derivs$, let $\Gamma(S) = \set{v| (\exists w \in S) \set{v,w} \in E}$ denote the set of \emph{neighbors} of~$S$.
We often write $\Gamma(v)$ for $\Gamma(\set{v})$.
The \emph{degree} of a vertex~$v$ is $|\Gamma(v)|$.  The graph is \emph{$d$-left-regular} if all left vertices have degree~$d$, and similarly for right-regular.
The graph is $(\leftdeg,\rightdeg)$-biregular if it is $\leftdeg$-left-regular and $\rightdeg$-right-regular.
\end{definition}

The vertices $\assets$ and $\derivs$ correspond to the assets and derivatives, respectively,
with an edge between a derivative vertex and asset vertex if the derivative depends on the asset.

Since Arora et al.\ showed how dense subgraphs can be problematic, we choose a graph with no dense subgraphs.
It is natural to use known constructions of suitable ``randomness extractors," which can be shown to lack dense subgraphs.
Indeed, this was our original approach.  However, we obtain stronger results in a simpler manner by considering the related \emph{expander graphs},
where we require expansion of asset vertices.

\begin{definition}
A bipartite graph on $[\leftsize] \cup [\rightsize]$ is an $(\kmax,\expansion)$-expander if for every subset $S \subseteq [\leftsize]$ of size at most $\kmax$,
\mbox{$|\Gamma(S)| \geq \expansion |S|$.}
\end{definition}

Note that we only need expansion of left vertices; expansion of right vertices is not required.
We will need a strong form of an expander, called a \emph{unique-neighbor expander}.

\begin{definition}
Let $\Gamma_i(S)$ denote the set of vertices $v \in \Gamma(S)$ with $|\Gamma(v) \cap S| = i$.
$\Gamma_1(S)$ are called the \emph{unique neighbors} of $S$.
\end{definition}

\begin{definition}
A bipartite graph on $[\leftsize] \cup [\rightsize]$ is an $(\kmax,\expansion)$-unique-neighbor expander if for every subset $S \subseteq [\leftsize]$ of size at most $\kmax$, $|\Gamma_1(S)| \geq \expansion |S|$.
\end{definition}

Note that to obtain unique neighbor expansion, the graph left-degree can't be too large.
Specifically, we must have $\kmax d < \rightsize$.
In other words, each asset participates in somewhat few derivatives.
This seems natural enough, although in the Future Work section we discuss trying to handle the case when this is false.

The following simple lemma is well known.

\begin{lemma}
\label{simple-unique}
A $d$-left-regular $(\kmax,\leftdeg - \Delta)$-expander is an $(\kmax,\leftdeg - 2\Delta)$-unique neighbor expander.
\end{lemma}

\begin{proof}
Consider any subset $S$ on the left of size $\ksize \leq \kmax$.  It has at least $(\leftdeg - \Delta)\ksize$ neighbors, which leaves at most $\Delta \ksize$
edges unaccounted for.  Thus 
$|\Gamma_1(S)| \geq |\Gamma(S)| - \Delta \ksize$,
as required.
\end{proof}

It is well known that most graphs are excellent expanders, which can be proven using the probabilistic method.
However, we need to be able to certify that a graph is an expander. 
It appears hard to do this for arbitrary graphs, which is related to
 Arora et al.\ impossibility results.
However, we can construct explicit expanders that are quite strong, though not as good as the non-explicit expanders for our purposes.

Explicit expander constructions are highly nontrivial.
The classic constructions of Margulis \cite{Mar:first,Mar:Ramanujan}, Gabber and Galil \cite{GabG}, and Lubotzky-Phillips-Sarnak \cite{lps} are not known to give unique-neighbor expanders.
Ta-Shma, Umans, and Zuckerman constructed the first unique-neighbor expanders of polylogarithmic left degree \cite{tuz},
and Capalbo et al.\ were the first to achieve constant left degree \cite{crvw}.
For our purposes, the best expanders were constructed by Guruswami, Umans, and Vadhan \cite{guv}, although these have polylogarithmic degree.
For more on expanders we refer the reader to the excellent survey \cite{hlw}.

\section{Expanders Give Pseudorandom CDOs}
\label{sec:expand-cdo}

Before discussing expander constructions, we first show how unique-neighbor expanders give pseudorandom CDOs.
It is helpful to compare our bounds to a natural trivial bound.  To this end, observe
that any biregular \cdofam\ is pseudorandom against $\ell$ lemons for $[a,b]$ tranches with error at most $\degree \lemonsize \delta/(\derivsize (b-a))$.
(Recall that $\delta$ is the difference between the expected values of a good asset and lemon.)
This is because converting $\ell$ good assets to lemons decreases the value of the entire CDO family by $\degree \ell \delta$,
since each lemon is in $\degree$ CDOs.

We show that a CDO family built from a $(\leftdeg,\rightdeg)$-biregular $(\lemonsize,\leftdeg - \Delta)$-unique neighbor expander
has error at most $\Delta \lemonsize \delta/(\derivsize (b-a))$.  That is, we replace $\degree$ from the trivial bound by $\Delta$.
Moreover, the naive bound on the error for the entire CDO is the maximum of the errors for each tranche.  We are instead able to improve the error to
$2\Delta \lemonsize \delta/(\derivsize \derivdeg)$.

The intuition for the proof is natural.  We consider some placement of lemons.  By the unique-neighbor expansion, we have fairly tight bounds on both
the number of derivatives containing no lemons, and the number containing exactly one lemon.  Thus, when we subtract values for two different lemon
placements, there is a lot of cancellation.

\begin{theorem}
\label{unique-cdo}
A CDO built from a $(\leftdeg,\rightdeg)$-biregular $(\lemonsize,\leftdeg - \Delta)$-unique neighbor expander is
pseudorandom for $\lemonsize$ lemons.
For the tranche $[a,b]$, the error is at most $\Delta \lemonsize \delta/(\derivsize(b-a))$, and for the entire CDO the error is
at most $2\Delta \lemonsize \delta/(\derivsize \derivdeg)$.
\end{theorem}

Before beginning the proof, we recall that $\val$ is the tranche value as defined in Definition~\ref{cdo-def}, and define the following.

\begin{definition}
Let $\valu_{[a,b]}(g) = \expect[\val_{[a,b]}(X)]$, where the random variable $X$ is the payoff of an underlying portfolio on $\derivdeg$ assets,
$g$ of which are good.
If the tranche is understood, we often omit it as a subscript.
\end{definition}

Since good assets first-order stochastically dominate lemons, we deduce that $\valu$ is a nondecreasing function of $g$.  This is obvious if lemons always take value zero, but
requires a short proof in general.

\begin{lemma}
\label{nondecreasing}
For any tranche $[a,b]$ and $g' \geq g$, we have $\valu_{[a,b]}(g') \geq \valu_{[a,b]}(g)$.
\end{lemma}

\begin{proof}
First fix $\globe = z$.  Now let $F_D$ denote the cumulative distribution function of distribution $D$.  We can choose random variables $X$ and $Y$ according
to  $D_g=D_g(z)$
and $D_\ell = D_\ell(z)$, respectively, by choosing $W \in [0,1]$ uniformly and outputting $X= F_{D_g}^{-1}(W)$ and $Y= F_{D_\ell}^{-1}(W)$.
This ``coupling" and
the domination condition imply that for every point in the probability space, $X \geq Y$.  Thus, we may substitute good assets for lemons in such a way that
for any point in the probability space, the value of every asset either increases or remains the same.  The lemma follows.
\end{proof}

Recall that $\mu$ and $\lambda$ are the expected values of each good asset and lemon, respectively, and $\delta = \mu - \lambda$.
Since a CDO simply restructures payoffs, the sum of the expected payoffs of the CDO equals the sum of the payoffs of the underlying assets,
implying the following observation.

\begin{observation}
\label{valuesum}
For any $g$, we have $\sum_{i=1}^s \valu_{[a_{i-1},a_i]}(g) = g\mu + (\derivdeg - g)\lambda = \derivdeg \lambda + g \delta$.
\end{observation}

Lemma~\ref{nondecreasing} and Observation~\ref{valuesum} imply the following corollary.

\begin{corollary}
\label{valuesumcor}
For any $g$, $i$, and tranche $[a,b]$, we have
$0 \leq \valu_{[a,b]}(g+i) - \valu_{[a,b]}(g) \leq i \delta$.
\end{corollary}

Let $\badthresh_i(L) = |\Gamma_i(L)|$, for $0 \leq i \leq \rightdeg$.
The following lemma sets up an expression for the error.

\begin{lemma}
\label{valuediff}
Fix a tranche $[a,b]$ and any $L,L' \subseteq [\assetsize]$ with $|L| = |L'| = \lemonsize$. Then:
\[ \totval(L') - \totval(L) = \sum_{i=1}^\rightdeg (\badthresh_i(L) - \badthresh_i(L')) (\valu(\rightdeg) - \valu(\rightdeg - i)).\]
\end{lemma}

\begin{proof}
Since $\cup_{i=0}^\rightdeg \Gamma_i(L) = [\derivsize]$, we have $\sum_{i=0}^\rightdeg \badthresh_i(L) = \derivsize$.

Observe that
\[ \totval(L) = \sum_{i=0}^\rightdeg \badthresh_i(L) \valu(\rightdeg - i).\]
Using $\sum_{i=0}^\rightdeg (\badthresh_i(L) - \badthresh_i(L'))=0$,
we obtain
\begin{eqnarray*}
\totval(L') - \totval(L)
&=& \sum_{i=0}^\rightdeg (\badthresh_i(L) - \badthresh_i(L')) (- \valu(\rightdeg - i)) \\
&=& \sum_{i=0}^\rightdeg (\badthresh_i(L) - \badthresh_i(L')) (\valu(\rightdeg) - \valu(\rightdeg - i)).
\end{eqnarray*}
Observing that the first term in the sum is zero gives the lemma.
\end{proof}

The following inequality will be useful.

\begin{lemma}
\label{lem:inequality}
Suppose $u_i \in [-\beta,\beta]$, $v_i \in [0,\delta]$ for $i=1,2,\ldots,\rightdeg$, and that
\begin{eqnarray}
\sum_{i=1}^\rightdeg u_i &=& 0, \label{zerosum}\\
\sum_{i=2}^\rightdeg |u_i| &\le& \beta. \label{upperbound}
\end{eqnarray}
Then 
\[ \left|\sum_{i=1}^\rightdeg u_i v_i \right| \le \beta \delta. \]
\end{lemma}

Before proving this lemma, we show how the lemma implies the theorem.

\begin{proof}[Proof of Theorem~\ref{unique-cdo}, single tranche]
Fix the tranche $[a,b]$, and we now bound its error.
We apply Lemma~\ref{lem:inequality} with
\begin{eqnarray*}
u_i &=& i(\badthresh_i(L) - \badthresh_i(L')),\\
v_i &=& (\valu(\rightdeg)-\valu(\rightdeg - i))/i,\\
\beta &=& \Delta \lemonsize.\\
\end{eqnarray*}

First note that
\[
\sum_{i=1}^\rightdeg i\badthresh_i(L) = \leftdeg \lemonsize,
\]
since both sides count the number of edges incident to $L$.
Therefore Equation (\ref{zerosum}) is satisfied.
To see that Equation (\ref{upperbound}) and $|u_1| \le \beta$ are satisfied, observe that
$\badthresh_1(L), \badthresh_1(L') \geq (\leftdeg - \Delta) \lemonsize$, and so
\[ \sum_{i=2}^\rightdeg i\badthresh_i(L) \leq \Delta \lemonsize,\]
and similarly for $L'$.
Corollary~\ref{valuesumcor} shows that $0 \le v_i \le \delta$.
By Lemmas~\ref{valuediff} and~\ref{lem:inequality}, we conclude that
\[ |\totval(L') - \totval(L)| = \left|\sum_{i=1}^\rightdeg u_i v_i \right| \le \beta \delta = \Delta \lemonsize \delta.\]
Dividing by $\derivsize (b-a)$ gives the result for the $[a,b]$ tranche.
\end{proof}

We now prove the inequality.

\begin{proof}[Proof of Lemma~\ref{lem:inequality}]
Assume without loss of generality that $u_1 \ge 0$.  Using Equation~(\ref{zerosum}), we get
\[ \sum_{i:u_i \ge 0} u_i = \sum_{i:u_i < 0} |u_i| \le \sum_{i>1} |u_i| \le \beta.\]
Therefore,
\[ \sum_{i=1}^\rightdeg u_i v_i \le \sum_{i:u_i \ge 0} u_i v_i \le \max_i \set{v_i} \sum_{i:u_i \ge 0} u_i \le \delta \beta.\]
Similarly,
\[ \sum_{i=1}^\rightdeg u_i v_i \ge \sum_{i:u_i < 0} u_i v_i \ge \max_i \set{v_i} \sum_{i:u_i < 0} u_i \ge \delta (-\beta).\]
\end{proof}

To analyze the error for the entire CDO, we use the following generalization of Lemma~\ref{lem:inequality}.

\begin{lemma}
\label{lem:inequality-matrix}
Suppose $u_i \in [-\beta,\beta]$, $v_{ij} \in [0,\delta]$ for $i=1,2,\ldots,\rightdeg$, and that
\begin{eqnarray}
\sum_{i=1}^\rightdeg u_i &=& 0, \label{zerosum2}\\
\sum_{i=2}^\rightdeg |u_i| &\le& \beta. \label{upperbound2}\\
(\forall i) \hspace*{.2in} \sum_{j=1}^s v_{ij} &\le& \delta. \label{rowsum}
\end{eqnarray}
Then 
\[ \sum_{j=1}^s \left|\sum_{i=1}^\rightdeg u_i v_{ij} \right| \le 2\beta \delta. \]
\end{lemma}

Before proving this lemma, we complete the proof of the theorem.

\begin{proof}[Proof of Theorem~\ref{unique-cdo}, entire CDO]
We now apply Lemma~\ref{lem:inequality-matrix} with the same choices of $u_i$ and $\beta$ as before, and with
\[ v_{ij} = (\valu_{[a_{j-1},a_j]}(\rightdeg)-\valu_{[a_{j-1},a_j]}(\rightdeg - i))/i.\]

Observation~\ref{valuesum} implies that Equation~(\ref{rowsum}) is satisfied, and the rest of the assumptions of Lemma~\ref{lem:inequality-matrix}
are satisfied as before.
We therefore obtain:
\[ \| \vectotval(L') - \vectotval(L) \|_1 = \sum_{j=1}^s \left|\sum_{i=1}^\rightdeg u_i v_{ij} \right| \le 2\beta \delta = 2\Delta \lemonsize \delta.\]

Dividing by $\derivsize \derivdeg = \leftdeg \assetsize$ gives the required result.
\end{proof}

We now prove the more general inequality.

\begin{proof}[Proof of Lemma~\ref{lem:inequality-matrix}]
Observe that
\begin{eqnarray*}
\sum_{j=1}^s \left|\sum_{i=1}^\rightdeg u_i v_{ij} \right| &\le& \sum_{i,j} |u_i v_{ij}|\\
&=& \sum_{i:u_i \ge 0} u_i \sum_j v_{ij} + \sum_{i:u_i < 0} |u_i| \sum_j v_{ij}\\
&\le& \beta \delta + \beta \delta.
\end{eqnarray*}
The last inequality comes from the bounds proved in Lemma~\ref{lem:inequality} that
\[ \sum_{i:u_i \ge 0} u_i = \sum_{i:u_i < 0} |u_i| \le \sum_{i>1} |u_i| \le \beta.\]
\end{proof}

\section{Constructive Expanders and CDOs}
\label{sec:explicit}

As stated earlier, despite the fact that
almost all graphs have excellent expansion properties, it is difficult to certify this efficiently for arbitrary graphs.
We therefore use the best known explicit expanders to build our pseudorandom CDOs.
For our purposes, the best explicit expanders are those by Guruswami, Umans, and Vadhan \cite{guv}.
We use Lemma~\ref{simple-unique} to go from expansion close to the degree to unique neighbor expansion.
This section is mostly about choosing the right parameters, and modifying the above graphs to ensure that they are biregular.


\begin{theorem}
\label{cdo}
For any $\alpha \in (0,1]$
and positive integers $\leftsize,\rightsize,\degree,\rightdeg$ such that $\leftsize \degree = \rightsize \rightdeg$,
the following holds for $\Delta = 2(2 \degree)^\alpha (\log_\degree \leftsize)\log_\degree \rightsize$
and any positive integer 
$\kmax \leq (\Delta \rightsize/(8\degree^3))^\alpha$.
There is an explicit pseudorandom
$(\assetsize,\derivsize,\rightdeg)$-CDO family against $\lemonsize$ lemons, for all $\lemonsize \leq \lmax$.
For the tranche $[a,b]$, the error is at most $2\Delta \lemonsize \delta/(\derivsize(b-a))$, and for the entire CDO the error is
at most $4\Delta \lemonsize \delta/(\derivsize \derivdeg)$.
\end{theorem}

\newcommand{\error}{\eta}
\newcommand{\temp}{s}
To prove this, we use the strong and elegant expander construction of
Guruswami, Umans, and Vadhan \cite{guv}.
We will set parameters in a different order, so we use their Theorem 3.3, obtained before they set parameters.

\begin{theorem}
\label{guv}
\cite{guv}
For $h$ any positive integer, $q$ a power of 2, and $\leftsize$ and $\rightsize$ powers of $q$, there is an explicit construction of an
$(\kmax,q-\Delta)$ expander on $[\leftsize] \cup [\rightsize]$ with left degree $q$, $\kmax = h^{\log_q \rightsize-1}$,
and $\Delta = (h-1) (\log_q \leftsize-1)(\log_q \rightsize-1)$.
\end{theorem}

Before setting parameters, we need the following simple observation.

\begin{observation}
\label{removal}
Suppose we are given
an $(\kmax,\degree - \Delta)$ expander with left-degree~$\degree$.
If we remove any left vertices, and add any right vertices, the graph remains an $(\kmax,\degree - \Delta)$ expander.
If for each left vertex, we remove an arbitrary $\degree - \degree'$ edges, then the graph becomes
an $(\kmax,\degree' - \Delta)$ expander with left degree~$\degree'$.
\end{observation}

We now set parameters from Theorem~\ref{guv} as follows.

\begin{corollary}
\label{guv-corollary}
For any $\alpha \in (0,1]$
and positive integers $\leftsize$, $\rightsize$, and $\degree$,
there is an explicit construction of
an $(\kmax,\degree - \Delta)$ expander on $[\leftsize] \cup [\rightsize]$ with left degree $\degree$ for
$\kmax = (\rightsize/(4\degree^2))^\alpha$ and
$\Delta = (2 \degree)^\alpha (\log_\degree \leftsize)\log_\degree \rightsize$.
\end{corollary}

\begin{proof}
Let $q$ be the smallest power of 2 that is at least $\degree$.
Let $\leftsize'$ be the smallest power of $q$ at least $\leftsize$, and let $\rightsize'$ be the largest power of $q$ at most $\rightsize$.
By Observation~\ref{removal}, it suffices to construct a $(\kmax,q - \Delta)$ expander on $[\leftsize'] \cup [\rightsize']$ with left-degree $q$.
Set $h=\ceil{q^\alpha}$ and $\ell =\log_q \rightsize' = \floor{\log_q \rightsize}$,
so $q^\ell \leq \rightsize < q^{\ell + 1}$.
We use the expander constructed in Theorem~\ref{guv}.  It suffices to lower bound $\kmax$ and upper bound $\Delta$.  We get:
\[ \kmax \geq h^{\ell - 1} \geq q^{\alpha (\ell - 1)} > (\rightsize/q^2)^\alpha \geq (\rightsize/(4\degree^2))^\alpha,\] and
\begin{eqnarray*}
\Delta &\leq& (h-1) (\log_q \leftsize' - 1)(\log_q \rightsize' - 1)\\
&<& q^\alpha (\log_q \leftsize)\log_q \rightsize < (2\degree)^\alpha (\log_\degree \leftsize)\log_\degree \rightsize.
\end{eqnarray*}

This completes the proof.
\end{proof}

This and other known unique-neighbor expander constructions
give left-regular graphs.  However, we need the graph to be biregular.
We show how to convert a left-regular graph to biregular while increasing the left-degree only slightly, at the expense of increasing the number of right vertices.
The following extends a lemma \mbox{from \cite{glr}.}

\begin{lemma}
\label{biregular}
Suppose we are given
a $\degree_0$-left-regular $(\kmax,\expansion)$ expander on $[\leftsize] \cup [\rightsize_0]$,
and parameters $\rightsize,\degree,\rightdeg$ such that $\leftsize \degree = \rightsize \rightdeg$, $\degree_0 < \degree \leq \rightsize_0$, 
and $\rightsize \geq \rightsize_0 \degree / (\degree - \degree_0)$.
We can efficiently construct
a $(\degree,\rightdeg)$-biregular $(\kmax,\expansion)$ expander on $[\leftsize] \cup [\rightsize]$.
\end{lemma}

\begin{proof}
Let $\rightdeg_0=\leftsize \degree_0/\rightsize_0$ denote the original average right degree.
For any right node $v \in [\rightsize_0]$ of degree $\rightdeg_v > \rightdeg$, divide it into $\ceil{\rightdeg_v/\rightdeg}$ vertices,
where $\floor{\rightdeg_v/\rightdeg}$ have degree $\rightdeg$ and at most one has degree less than $\rightdeg$.
(Partition neighbors arbitrarily.)

The number of new nodes added is at most
\[ \sum_{v \in [\rightsize_0]} \left( \left \lceil \frac{\rightdeg_v}{\rightdeg} \right \rceil - 1 \right ) < \sum_{v \in [\rightsize_0]} \frac{\rightdeg_v}{\rightdeg} =
\frac{\rightsize_0\rightdeg_0}{\rightdeg} = \frac{\leftsize \degree_0 \rightsize}{\leftsize \degree} = \frac{\degree_0 \rightsize}{\degree}.
\]

Thus, the total number of right nodes is less than 
\[ \rightsize_0 + \frac{\degree_0}{\degree} \rightsize
\leq \frac{\degree - \degree_0}{\degree} \rightsize + \frac{\degree_0}{\degree} \rightsize = \rightsize.\]

Add isolated nodes to the right to make the total number of right nodes exactly $\rightsize$.
Now add edges arbitrarily to the right and left to make all left degree $\degree$ and right degrees $\rightdeg$, which is possible because
$\leftsize \degree = \rightsize \rightdeg$.  Naively, this may allow multiple edges, but we can avoid this by filling edge slots
in the following order.  For left nodes, cycle over all nodes $\degree - \degree_0$ times, filling one edge slot each time.
For right nodes, cycle over all nodes once, filling all edge slots for a node before proceeding to the next node.
\end{proof}

\begin{corollary}
\label{guv-corollary-biregular}
For any $\alpha \in (0,1]$
and positive integers $\leftsize$, $\rightsize$, $\degree$, and $\rightdeg$ such that $\leftsize \degree = \rightsize \rightdeg$,
there is an explicit construction of a $(\degree,\rightdeg)$-biregular
$(\kmax,\degree - \Delta)$ expander on $[\leftsize] \cup [\rightsize]$ for
$\Delta = 2(2 \degree)^\alpha (\log_\degree \leftsize)\log_\degree \rightsize$ and
$\kmax = (\Delta \rightsize/(8\degree^3))^\alpha$.
\end{corollary}

\begin{proof}
Set $\Delta_0 = \Delta/2$, $\degree_0 = \degree - \Delta_0$, and $\rightsize_0 = \Delta_0 \rightsize/\degree$.  By Corollary~\ref{guv-corollary},
there is an explicit construction of a $(\kmax,\degree_0 - \Delta_0')$ expander on $[\leftsize] \cup [\rightsize_0]$ with left degree $\degree_0$ for
$\kmax' = (\rightsize_0/(4\degree_0^2))^\alpha \geq \kmax$ and
$\Delta_0' = (2 \degree_0)^\alpha (\log_\degree \leftsize)\log_\degree \rightsize_0 \leq \Delta_0$.
Now apply Lemma~\ref{biregular}.
\end{proof}

Combining Corollary~\ref{guv-corollary-biregular} and Lemma~\ref{simple-unique} with Theorem~\ref{unique-cdo} yields Theorem~\ref{cdo}.

If $\degree$ is smaller, we could use the expanders of \cite{crvw}, but the degree is not as good a function in the error and our results are not as~strong.

We remark that our required notion of explicitness is weaker than one often considered in the literature and achieved in the above constructions.
We just need the whole graph to be efficiently constructible, whereas sometimes one needs the $i$th neighbor of a node to be computable very quickly,
say in time polylogarithmic in the number of nodes.

\section{Two General Assets}
\label{sec:general}

In this section we obtain bounds even if the probability distribution of lemons is not stochastically dominated by the probability distribution of good assets.
We only assume that $\mu \geq \lambda$, where $\mu$ and $\lambda$ are the expected values of each good asset and lemon, respectively.
We don't need to think of the second asset as a lemon; instead consider two general assets with expected values $\mu \geq \lambda$.
Of course, in this more general setting our results our weaker.
Now, even the ``trivial" bounds change; such bounds can be deduced from our bounds below.
We show that in the general case,
the $\delta=\mu-\lambda$ in Theorem~\ref{unique-cdo} must be replaced by $\mu$ for the $[a,b]$ tranche.  Moreover, we no longer get better bounds for
the entire CDO than can be deduced from the bounds in the tranches.
On the other hand, replacing $\delta$ by $\mu$ is not too bad if $\lambda$ is much smaller than $\mu$.

\begin{theorem}
\label{unique-cdo-gen}
A CDO built from a $(\leftdeg,\rightdeg)$-biregular $(\lemonsize,\leftdeg - \Delta)$-unique neighbor expander is
pseudorandom for $\lemonsize$ lemons.
For the tranche $[a,b]$, the error is at most $\Delta \lemonsize \mu/(\derivsize(b-a))$.
\end{theorem}

The only place we used stochastic domination was to prove
Lemma~\ref{nondecreasing}, and hence deduce Corollary~\ref{valuesumcor}.
Corollary~\ref{valuesumcor}  may no longer hold, but we prove an analog with $\delta$ replaced by $\mu$.

\begin{lemma}
\label{valuechange}
For any $g$, $i$, and tranche $[a,b]$, we have
$|\valu_{[a,b]}(g+i) - \valu_{[a,b]}(g)| \leq i \mu$.
\end{lemma}

To prove this, it is helpful to use the following expression for the value of the $[a,b]$ tranche.

\begin{observation}
\label{obs:value}
Let $W$ denote a random variable representing the payoff of a portfolio underlying a CDO. 
Then the value of the corresponding $[a,b]$ tranche is
\[ \int_a^b \Pr[W > w]dw.
\]
\end{observation}

\begin{proof}[Proof of Lemma~\ref{valuechange}]
It suffices to prove the lemma for $i=1$.
Let $W$ be the payoff of a portfolio on $\derivdeg - 1$ assets, $g$ of which are good.
Let $X$ be the payoff of a good asset, and $Y$ the payoff of a lemon.  We wish to show that
\[
\left | \int_a^b \Pr[W+X > u]du - \int_a^b \Pr[W+Y > u]du \right | \leq \mu.
\]

To this end, first note that
\[
\int_a^b \Pr[W+X > u]du \geq \int_a^b \Pr[W > u]du,
\]
and similarly for $W+Y$, since both $X$ and $Y$ are nonnegative.
It therefore suffices to show that
\[
\int_a^b \Pr[W+X > u]du - \int_a^b \Pr[W > u]du \leq \mu,
\]
and hence the corresponding inequality for $W+Y$.

Condition on $W=w$; we show this inequality for any $w$.  Observe that
\[ \Pr[w+X > u] - \Pr[w > u] = \left\{ \begin{array}{ll}
								\Pr[X > u-w]	& \mbox{if $w \leq u$}\\
								0			& \mbox{otherwise}
								
								\end{array}
							\right. \]
							
Letting $x=u-w$ gives:
\[
\int_a^b \left ( \Pr[w+X > u] - \Pr[w > u] \right ) du \leq \int_0^\infty \Pr[X>x] dx = \mu.
\]
This completes the proof.
\end{proof}

\begin{proof}[Proof of Theorem~\ref{unique-cdo-gen}]
We now proceed as in our earlier proof, replacing Corollary~\ref{valuesumcor} with Lemma~\ref{valuechange}.
Everything else goes through as before.
\end{proof}

\section{Extensions}
\label{sec:extend}

\subsection{Partial Stochastic Domination}

In the previous section 
we removed the stochastic domination assumption and obtained Theorem~\ref{unique-cdo-gen}, which has a weaker conclusion than Theorem~\ref{unique-cdo}.
Now we introduce a notion of partial stochastic domination that allows us to interpolate somewhat between Theorems~\ref{unique-cdo-gen} and~\ref{unique-cdo}.
For the following definition, when we write a distribution $D=pD_1 + (1-p)D_2$, we mean that we can sample from $D$ by sampling from $D_1$ with probability $p$ and sampling from $D_2$ otherwise.

\begin{definition}
We say a distribution $D$ \emph{$p$-dominates $D'$ with dominated mean $\lambda_1$} if $D=pD_1 + (1-p)D_2$, $D'=pD'_1 + (1-p)D'_2$, $D_1$ first-order stochastically dominates $D_1'$, and
the mean of $D'_1$ is $\lambda_1$.
\end{definition}

For example, a symmetric good asset distribution $1/2$-dominates a symmetric lemon distribution, as $D_1$ for the good asset occurs when the value is above its median, and $D_2$ for the lemon occurs when the value is below its median.  However, the dominated mean will be low with this pairing.
If the lemon distribution is a shift and spread of the good asset distribution, with the spread being a factor of $s>1$,
then we might expect the good asset distribution to $p$-dominate the lemon distribution for $p \approx 1/s$,
and with dominated mean not far from $\lambda$ if $s \approx 1$.

We can show the following.

\begin{theorem}
\label{interpolate}
Suppose the distribution of good assets $p$-dominates the distribution of lemons with dominated mean~$\lambda_1$.
A CDO built from a $(\leftdeg,\rightdeg)$-biregular $(\lemonsize,\leftdeg - \Delta)$-unique neighbor expander is
pseudorandom for $\lemonsize$ lemons.
For the tranche $[a,b]$, the error is at most $\Delta \lemonsize \sigma/(\derivsize(b-a))$,
where $\sigma=\mu - p \lambda_1$.
\end{theorem}

Note that Theorems~\ref{unique-cdo} and~\ref{unique-cdo-gen} follow from this, by taking $p=1$ (in which case $\lambda_1=\lambda$) and $p=0$, respectively.

To motivate the proof of this theorem, observe
that the only place we used stochastic domination to show Theorem~\ref{unique-cdo} was to prove
Lemma~\ref{nondecreasing}, and hence deduce Corollary~\ref{valuesumcor}.
The proof of Theorem~\ref{interpolate} follows from the following interpolation between Corollary~\ref{valuesumcor} and Lemma~\ref{valuechange}.

\begin{lemma}
\label{valuechange-interpolate}
For any $g$, $i$, and tranche $[a,b]$, we have
$|\valu_{[a,b]}(g+i) - \valu_{[a,b]}(g)| \leq i \sigma$.
\end{lemma}

We sketch a proof of this lemma for $i=1$, which suffices.  Write the good asset distribution $D^g = pD^g_1 + (1-p)D^g_2$ and the lemon distribution $D^\ell = pD^\ell_1 + (1-p)D^\ell_2$,
with $D^g_1$ dominating $D^\ell_1$, and the means of $D^g_1$, $D^g_2$, $D^\ell_1$, and $D^\ell_2$ being $\mu_1$, $\mu_2$, $\lambda_1$, and $\lambda_2$, respectively.  Couple the random variables corresponding to a lemon and a good asset in the natural way, so that with probability $p$ the sample is chosen from $D^g_1$ and $D^\ell_1$.  The expected increase in value by switching from a lemon to a good asset is the probability that the domination condition holds, i.e., $p$, times the expected increase when domination holds, which is bounded in Corollary~\ref{valuesumcor} and proportional to $\mu_1-\lambda_1$, plus the probability that domination isn't guaranteed to hold, i.e., $1-p$, times the expected increase in the general case, which is bounded in Lemma~\ref{valuechange} and proportional to $\mu_2$.  Hence the total increase is proportional to $p(\mu_1-\lambda_1) + (1-p)\mu_2 = \mu - p\lambda_1 = \sigma$.
The theorem then follows as before.

\subsection{Allowing Different Distributions}

It's also natural to ask whether we can analyze a model where the good assets (respectively, lemons) don't all have the same distribution.  For this discussion, let's fix the global random variable $Z$.  What if all the distributions can be different, but each good asset first-order stochastically dominates each lemon?
This case can arise in practice because banks may use different assets that have different risk profiles.
Unlike for lemons, the risk profiles of the good assets are public information.

Unfortunately, stochastic domination is not a strong enough assumption.  For example, suppose there are few lemons, each with value 0, but the good assets are divided into two categories, excellent and lemon-like.  The lemon-like good assets are only slightly better than lemons, and the excellent assets are genuine good assets.  This satisfies the domination conditions.  However, the error should be proportional to the total number of lemons plus lemon-like good assets, not the number of lemons as we would have liked.

This counterexample demonstrates that the good assets should all have similar distributions, as should all the lemons.  It is not enough that the expected values are similar, because an asset that always takes value 1/2 behaves very differently in a CDO from an asset that takes value 1 with probability 1/2 and 0 otherwise.  A natural measure of closeness in our setting is the Levy distance.

\begin{definition}
Let $X$ and $Y$ be real-valued random variables with cumulative distribution function $F$ and~$G$.
The Levy distance between them is the smallest $\alpha$ such that for all $x$,
\[ F(x-\alpha )-\alpha \leq G(x)\leq F(x+\alpha )+\alpha.
\]
\end{definition}

Consider the model where any two lemons (respectively, good assets) have Levy distance at most $\alpha$.
If two assets have Levy distance $\alpha$, then a CDO containing one can change in value by at most $2\alpha$ if it is swapped out for the other.
Therefore, the total value of a tranche can change by at most $2mr \alpha$, so the error can increase by at most $\frac{2mr \alpha}{m(b-a)} =2r\alpha/(b-a)$.

\section{Conclusions and Future Work}
\label{sec:conclude}

We used expander graphs to construct pseudorandom derivative families, whose values cannot be manipulated by strategic lemon placement.
We analyzed our construction under a fairly general model with two asset types, where good assets first-order stochastically dominating lemons.

In the previous section, we discussed extensions, but perhaps more is possible.
Can we more precisely analyze the case if the lemon distribution is a shift and spread of the good asset distribution?
Is there a different useful condition besides (partial) stochastic domination?
What else can we say when the distributions of good assets (or lemons) differ?

Another direction is to allow more general dependencies than that given by the factor framework.

Other future work can involve the graph-theoretic aspects.
Unique-neighbor expansion is fairly strong; can we obtain good bounds under a weaker condition?
For example, perhaps a seller wishes to place each asset in more derivatives, so the graph is denser.
Then it's natural to require the graph to be a ``randomness extractor."  Can we obtain bounds for such graphs?

What if the graph is not left-regular?  This corresponds to some assets appearing in fewer CDOs than others, which is natural if they are ``smaller."

Finally, we could try to analyze a more complex derivative that banks created and sold:  CDOs squared.
A CDO squared is a CDO whose underlying assets are themselves tranches of CDOs.

\section*{Acknowledgements}

I thank Sanjeev Arora, 
Rafa Mendoza-Arriaga,
Abhishek Bhowmick,
Michael Kearns,
Kumar Muthuraman,
Ryan O'Donnell, and
Stathis Tompaidis
for useful comments and discussions.
I'm also grateful to the anonymous referees for many valuable comments.

This research was supported by NSF Grants CCF-0634811, CCF-0916160, CCF-1526952, and CCF-1705028
and a Simons Investigator Award (\#409864).

\bibliographystyle{alpha}
\bibliography{/Users/diz/Documents/bibs/refs}

\end{document}